\documentclass{llncs}

\usepackage[T1]{fontenc}
\usepackage[utf8]{inputenc}

\usepackage{amsmath}
\usepackage{amssymb}
\usepackage{fdsymbol}

\usepackage{amsthm}

\usepackage{bm}
\usepackage{csquotes}
\usepackage{enumitem}
\usepackage{mathtools}
\usepackage[normalem]{ulem}
\usepackage{color, colortbl}
\usepackage[ruled, vlined, linesnumbered]{algorithm2e}
\usepackage{algpseudocode}
\usepackage{multirow}
\usepackage{cite}
\usepackage{placeins}
\usepackage{enumitem}
\usepackage{scalerel}

\SetKwInput{KwData}{Input}
\SetKwInput{KwResult}{Output}

\setlist[itemize]{parsep=1pt, topsep=1pt}
\newcommand\FF{\mathbb{F}}
\newcommand\ZZ{\mathbb{Z}}

\newcommand\calA{\mathcal{A}}
\newcommand\calB{\mathcal{B}}
\newcommand\calC{\mathcal{C}}
\newcommand\calD{\mathcal{D}}

\newcommand\calF{\mathcal{F}}
\newcommand\calG{\mathcal{G}}
\newcommand\calH{\mathcal{H}}
\newcommand\calI{\mathcal{I}}
\newcommand\calJ{\mathcal{J}}

\newcommand\calL{\mathcal{L}}
\newcommand\calM{\mathcal{M}}

\newcommand\calP{\mathcal{P}}

\newcommand\calS{\mathcal{S}}

\newcommand\calV{\mathcal{V}}

\newcommand\bfa{{\bm{a}}}

\newcommand\bfc{{\bm{c}}}

\newcommand\bfe{{\bm{e}}}

\newcommand\bfg{{\bm{g}}}
\newcommand\bfh{{\bm{h}}}

\newcommand\bfm{{\bm{m}}}

\newcommand\bfu{{\bm{u}}}

\newcommand\bfy{{\bm{y}}}

\newcommand\bfA{{\bm{A}}}
\newcommand\bfB{{\bm{B}}}

\newcommand\bfG{{\bm{G}}}

\newcommand\bfI{{\bm{I}}}

\newcommand\bfP{{\bm{P}}}

\def\sq{\mathbin{\scalerel*{\strut\rule{4ex}{4ex}}{\cdot}}}

\newcommand\gpub{\bm{G}_{\text{pub}}}
\newcommand\cpub{\mathcal{C}_{\text{pub}}}
\newcommand\hpub{\bm{H}_{\text{pub}}}
\newcommand\hsec{\bm{H}_{\text{sec}}}

\newcommand\cpubp{{\cpub^{\perp}}}
\DeclareMathOperator{\rk}{rk}

\renewcommand\epsilon{\varepsilon}

\title{An analysis of Coggia-Couvreur Attack on Loidreau’s
Rank-metric public-key encryption scheme in the general case}

\author{Pierre Loidreau\inst{1} \and  Ba-Duc Pham\inst{2}}
\institute{Univ Rennes, DGA MI, CNRS, IRMAR - UMR 6625, F-35000 Rennes, France \\\email{pierre.loidreau@univ-rennes1.fr } \and Univ Rennes, IRMAR - UMR 6625, F-35000 Rennes, France \\\email{ ba-duc.pham@univ-rennes1.fr}
}

\begin{document}
\maketitle \thispagestyle{plain} \pagestyle{plain}

\begin{abstract}
In this paper we show that in the case where the public-key can be distinguished from a random code in Loidreau’s encryption scheme, then Coggia-Couvreur attack can be extended to recover an equivalent secret key. This attack can be conducted in polynomial-time if the masking vector space has dimension $3$, thus recovering the results of Ghatak. 
\end{abstract}


\section*{Introduction}
Since the use of $\FF_{q^m}$-linear rank metric permits to design a short public key encryption scheme, one of the directions of code based cryptography consists in instantiating McEliece encryption scheme \cite{Eliece78} with codes in rank metric, \cite{gpt91, gmrz13}.

Because of the structure of Gabidulin codes, any cryptosystem instantiated with codes containing Gabidulin codes not sufficiently scrambled was attacked \cite{over08}. In 2017, Loidreau proposed a scheme based on  Gabidulin codes masked with a small dimensional vector space \cite{Loi17}. If the dimension of the vector space is too small, then there exists a very simple polynomial-time distinguishing algorithm. 

The question was to know if distinguishing is enough to break. Coggia and Couvreur \cite{CC19} showed that in the case where the dimension of the masking space is $2$, a decryption procedure can be recovered in polynomial-time.  More recently, Ghatak  \cite{gha20} claimed that this approach could be extended to a masking space of dimension $3$. 

In this work we show that this can be extended to any dimension. The attack is not necessary polynomial, but we include the previous results. Moreover we are able to prove rigorously under some assumptions the efficiency of the attack. 

\subsection*{Contribution}

\begin{enumerate}
    \item Completing the key-recovery attack for any $\lambda$
    \item In \cite{gha20}, the author uses one reduced polynomial to determine $\gamma = \{\gamma_1, \gamma_2\}$ where $\gamma$ specifies the secret subspace but their assumption is not clear in practice when we implement in MAGMA. In this paper, the use of system of polynomial equation gives a proof for the equivalent between the set of roots and the orbit of one root under the action of \textbf{PGL}$(3,\FF_q)$. It completes the polynomial time key recovery attack.

\end{enumerate}

\subsubsection*{Organization of the article} The first section outlines Loidreau’s scheme. In the beginning of the next section, we
formulate the distinguisher for any dimension of the secret subspace. Afterwards, we describe the attack for any $\lambda$ in 4 steps. In the end of this section, we analyse the complexity of the attack in case $\lambda=3$. We conclude with a discussion on the results and future works.

\section{The encryption scheme}
\label{Sec:Scheme}	

\subsection{Generalities}

Let $\bfG$ a random generator matrix of a Gabidulin code $\calG_k(\bfg)$. Fix an integer $\lambda \leq m$ and an $\FF_q$-vector subspace $\calV$ of $\FF_{q^m}$ of dimension $\lambda$. Let $\bfP \in \bm{GL}(n,\FF_{q^m})$ whose entries are all in $\calV$. Then, let 
\[
    \gpub = \bfG\bfP^{-1}
\]

\begin{itemize}
    \item KeyGen: Public key ($\gpub,t$) where $t = \lfloor \frac{n-k}{2\lambda}\rfloor$
    
     \hspace{1.5 cm} Secret key ($\bfg,\bfP$)
     \item Encryption: Given a plaintext $\bfm \in \FF_{q^m}^k$, choose $e \in \FF_{q^m}^n$ of rank weight $t$. The ciphertext is:
     
     \[
         \bfc = \bfm\gpub+ \bfe
     \]
     
     \item Decryption:
     	\begin{itemize}
     	\item Compute $\bfc\bfP = \bfm\bfG + \bfe\bfP$.
     
        \item  Decode in $\calG_k(\bfg)$ and $\rk(\bfe\bfP) \leq t\lambda \leq \frac{n-k}{2}$
\end{itemize}     	 
     	
\end{itemize}

Let us denote by $\cpub$ the code generated by  $\gpub$ and  by $\cpub^\perp$, the dual code. Let $\hpub$ be a generator matrix of $\cpub^\perp$.  It is immediate that 
 \[
        \hpub = \hsec\bfP^T
    \]
    where $\hsec$ is a parity-check matrix of $\calG_k(\bfg)$.

\subsection{Goal of a reconstructing attack and solution set}
\label{Sec:Solutions}

Our main goal is to design a reconstructing attack from the knowledge of $\cpub^\perp$
and under some particular sets of parameters. 

W.l.o.g, one can suppose that $1\in \calV$. Suppose that $\calV = \left\langle 1,\beta_1,\dots,\beta_{\lambda-1}\right\rangle_{\FF_q}$ for some $\{\beta_i\}_{i=1}^{\lambda-1} \in \FF_{q^m} \backslash \FF_q$. Therefore, $\bfP^T$ can be decomposed into

\[
    \bfP^T = \bfP_0 + \sum\limits_{i=1}^{\lambda-1}\beta_i \bfP_i
\]
where $\bfP_i$ are $n \times n$ matrices with entries in $\FF_q$ not necessarily invertible.

Let $\calC_{\text{sec}}^\perp$ the dual code of $\calG_k(\bfg)$. Thus, $\calC_{\text{sec}}^\perp = \calG_{n-k}(\bfa)$ for some $\bfa \in \FF_{q^m}^n$ with $\rk(\bfa) = n$. We define
\[
    \bfh_0 = \bfa\bfP_0, \bfh_1 = \bfa\bfP_1, \dots, \bfh_{\lambda-1} = \bfa\bfP_{\lambda-1}
\]

\begin{lemma}
    The code $\cpub^\perp$ is spanned by $\bfh_0^{[i]} + \sum\limits_{j=1}^{\lambda-1}\beta_j\bfh_j^{[i]}$ for $i = 0,\dots,n-k-1$
\end{lemma}

\begin{proof}
    For any $\bfc \in \cpub^\perp$, there exists $P \in \FF_{q^m}[X;\theta]$ of degree smaller than $n-k$ such that
    \[
        \bfc = P\langle\bfa\rangle\bfP^T = P\langle\bfa\rangle\bfP_0 + \sum\limits_{i=1}^{\lambda-1}\beta_i P\langle\bfa\rangle\bfP_i = P\langle\bfh_0\rangle + \sum\limits_{i=1}^{\lambda-1}\beta_iP\langle\bfh_i\rangle
    \]
\end{proof}

\noindent Let us define the so-called solution set of the encryption scheme

\begin{definition}[Solution set]
\label{Defi:Solutions}
 The set $\mathcal{S}$  of all  $(\bfh,\vec{\beta}) \in (\FF_{q^m}^n)^{\lambda} \times \FF_{q^m}^{\lambda-1}$ such that 
 \begin{equation}
\label{Eq:Attack}
    \cpubp = \left\langle \bfh_0^{[i]} + \sum\limits_{j=1}^{\lambda-1}\beta_j\bfh_j^{[i]},   ~ i = 0,\dots,n-k-1\right\rangle
\end{equation}
where $\forall j = 0,\dots, \lambda, ~\bfh_j$ has rank $n$ and   $\left\langle 1,\beta_1,\ldots,\beta_{\lambda-1} \right\rangle_{\FF_q}$ has dimension $\lambda$ is called solution set of the encryption scheme.
\end{definition}

It is obvious that finding an element of the solution set $\mathcal{S}$ implies the ability to design a polynomial-time decryption algorithm. What we call a reconstructing attack corresponds to finding  an element in  $\mathcal{S}$. The solution set   $\mathcal{S}$ has the following properties.    

\begin{proposition}
\label{prop:solutions}
Let $(\bfh,\vec{\beta}) \in (\FF_{q^m}^n)^{\lambda} \times \FF_{q^m}^{\lambda-1}$. Let $\bfA=  (a_{j,i})_{j,i=0}^{\lambda-1} \in GL_{\lambda}(\FF_{q})$.
Let us define the following group action on $(\FF_{q^m}^n)^{\lambda} \times \FF_{q^m}^{\lambda-1}$ by $\bfA\cdot (\bfh,\vec{\beta}) = (\bfh',\vec{\beta'})$  where 
\[
\left\{
\begin{array}{ll}
\bfh_j = \dfrac{a_{j,0}\bfh'_0 + \sum\limits_{i=1}^{\lambda-1}a_{j,i}\bfh'_i}{a_{0,0} + \sum\limits_{i=1}^{\lambda-1}a_{i,0}\beta_i},  ~j = 0,\dots, \lambda-1 &\\
\beta'_j = \dfrac{a_{0,j} +\sum\limits_{i=1}^{\lambda-1}a_{i,j}\beta_i}{a_{0,0} + \sum\limits_{i=1}^{\lambda-1}a_{i,0}\beta_i}, ~j = 1,\dots, \lambda -1 &
\end{array}
   \right.
\]
\begin{enumerate}
\item Then if 
$(\bfh, \vec{\beta}) \in \mathcal{S}$  we have   $ \bfA\cdot (\bfh, \vec{\beta}) \in \mathcal{S}$ 
\item Moreover let 
$\overline{\bfA} = \{\bfB \in GL_\lambda(\FF_q) | \exists c \in \FF_q^*, \bfB = c\bfA \}$ . 
Then, for any $\bfB \in \overline{\bfA}$, and for any $(\bfh, \vec{\beta}) \in (\FF_{q^m}^n)^{\lambda} \times \FF_{q^m}^{\lambda-1}$ we have

\[
\bfA\cdot (\bfh,\vec{\beta}) = \bfB \cdot (\bfh,\vec{\beta})
\]

\end{enumerate}

\end{proposition}

\begin{proof}
Let $(\bfh,\vec{\beta}) \in \mathcal{S}$. 
Since from the definition of  $\mathcal{S}$ the elements $1,\beta_1,\ldots,\beta_{\lambda-1}$ are $\FF_q$-linearly independent, and since $\bfA$ is non singular, $(a_{0,0},\ldots, a_{0,\lambda-1}) \ne 0$
 this implies that $a_{0,0} + \sum\limits_{i=1}^{\lambda-1}a_{0,j}\beta_i \ne 0$.
 Therefore, the elements $\beta'_1,\ldots,\beta'_{\lambda-1}$ are well defined and for all  $0 \le \ell \le n-k-1$,
    
    \begin{equation*}
        \begin{aligned}
            &{\bfh'}_0^{[\ell]} + \sum\limits_{j=1}^{\lambda-1}\beta'_j{\bfh'}_j^{[\ell]}\\
            =& \dfrac{1}{  a_{0,0} + \sum\limits_{i=1}^{\lambda-1}a_{i,0}\beta_i}\left(\left(   a_{0,0} + \sum\limits_{i=1}^{\lambda-1}a_{i,0}\beta_i    \right){\bfh'}_0^{[\ell]} + \sum\limits_{j=1}^{\lambda-1}\left(         
            a_{0,j} +\sum\limits_{i=1}^{\lambda-1}a_{i,j}\beta_i \right){\bfh'}_j^{[\ell]}\right)\\            
            =& \dfrac{1}{  a_{0,0} + \sum\limits_{i=1}^{\lambda-1}a_{i,0}\beta_i}\left(\left(a_{0,0}\bfh'_0 + \sum\limits_{j=1}^{\lambda-1}a_{0,j}\bfh'_j\right)^{[\ell]} + \sum\limits_{i=1}^{\lambda-1}\beta_i\left(a_{i,0}\bfh'_0 + \sum\limits_{j=1}^{\lambda-1}a_{i,j}\bfh'_j\right)^{[\ell]}\right)\\            
            =& \left(a_{0,0} + \sum\limits_{i=1}^{\lambda-1}a_{i,0}\beta_i\right)^{[\ell]-1}\left(\bfh_0^{[\ell]} + \sum\limits_{i=1}^{\lambda-1}\beta_i\bfh_i^{[\ell]} \right)
        \end{aligned}
    \end{equation*}
    Therefore  
    \[ 
        \cpubp \stackrel{def}{=}   \left\langle \bfh_0^{[\ell]} + \sum\limits_{i=1}^{\lambda-1}\beta_i\bfh_i^{[\ell]}, ~ 0 \le \ell \le n-k-1 \right\rangle  = \left\langle {\bfh'}_0^{[\ell]} + \sum\limits_{i=1}^{\lambda-1}\beta'_i{\bfh'}_i^{[\ell]}, ~ 0 \le \ell \le n-k-1 \right\rangle
          \]
   Thus,  $(\bfh',\vec{\beta'}) \in \calS$.

For the second point of the proposition : Let $\bfB \in \overline{\bfA}$. It means that there exists $c \in \FF_q^*$ such that $\bfB\bfA^{-1} = c\bfI$ or equivalently $\bfB = c\bfA$.

Let $(\bfh'',\vec{\beta''}) \stackrel{def}{=} \bfB \cdot (\bfh',\vec{\beta'})$ , where $ (\bfh',\vec{\beta'}) = \bfA\cdot (\bfh,\vec{\beta})\in (\FF_{q^m}^n)^{\lambda} \times \FF_{q^m}^{\lambda-1}$. It means that

\[
\left\{
\begin{array}{l}
\bfh_j' = \dfrac{b_{j,0}\bfh''_0 + \sum\limits_{i=1}^{\lambda-1}b_{j,i}\bfh''_i}{b_{0,0} + \sum\limits_{i=1}^{\lambda-1}b_{i,0}\beta_i},  ~j = 0,\dots, \lambda-1 \\
\beta''_i = \dfrac{b_{0,j} +\sum\limits_{i=1}^{\lambda-1}b_{i,j}\beta_i}{b_{0,0} + \sum\limits_{i=1}^{\lambda-1}b_{i,0}\beta_i}, ~j = 1,\dots, \lambda -1
\end{array}
   \right.
\]

It is obvious that the actions of 2 matrix $\bfA,\bfB$ give us same image, which is $(\bfh'',\vec{\beta''}) = (\bfh',\vec{\beta'})$.

\end{proof}

\section{Attacks on the system}

Now the section is organised as follows: In a first part we make a brief summary of the attack. Since it is technical, this section highlights the different principles. In first subsection, we introduce the distinguisher between the dual of the public code and the random codes. Afterwards, we make some assumptions in the beginning and under these assumptions, we exploit the attack based on the distinguisher introduced in 2.1. We also consider the special case of $\lambda=3$ and analyze its complexity in Section 2.3.

We will also need to introduce a special setting to simplify the technicalities of the proofs. For any $(\bfh,\vec{\beta}) \in \mathcal{S}$, we denote  
\[
\bfy_{(\bfh,\vec{\beta})}^{[u,j]} = \bfh_0^{[j]} + \sum\limits_{i=1}^{\lambda-1} \beta_i^{[u]}\bfh_i^{[j]}
\]
for any integers $(u,j)$. For a given $(\bfh,\vec{\beta})$, to simplify, we denote it by $\bfy^{[u,j]}$. Moreover, we can denote $\bfy^{[M]} = \{\bfy^{[u,j]}, (u,j) \in M \subset \mathbb{Z} \times \mathbb{Z}\}$. The codes that we will consider will be generated by the $\bfy^{[u,j]}$, where $(u,j) \in \mathbb{Z} \times \mathbb{Z}$. 
Let $I \subset  \mathbb{Z} \times \mathbb{Z}$. We denote by 
\[
\calC_I \stackrel{def}{=} \left\langle \bfy^{[u,j]}, (u,j) \in I \right\rangle,
\]
if $I$ is non-empty and $\calC_{\emptyset} \stackrel{def}{=} \{0\}$.
From the expression of $\cpubp$ in Definition \ref{Defi:Solutions}, we have
\begin{equation}
\label{Eq:CpubIndices}
\cpubp = \calC_{\{0\}\times [0,\ldots,n-k-1]}
\end{equation}

Now we introduce a very fundamental theorem which will support all of our future proofs
\begin{theorem}[CodeSet theorem]
\label{Theo:CodeSet}
 We have 
\[
\forall I,~J \subset  \mathbb{Z}\times \mathbb{Z},~ 
\left\{ 
\begin{array}{l}
\calC_{I \cup J} = \calC_I + \calC_J \\
\calC_{I \cap J} \subset \calC_I \cap \calC_J 
\end{array}
\right.
\]
If moreover $M \subset \mathbb{Z}\times \mathbb{Z}$, where  $\bfy^{[u,j]},~ (u,j) \in M$ are $\FF_q^m$  linearly independent, then
\begin{itemize}
\item for all $I \subset M, ~\dim(\calC_I) = |I|$
\item  $\forall I,~J \subset M,~ \calC_I \cap \calC_J = \calC_{I \cap J}$
\item $\forall I,~J \subset M, ~ \calC_{I \sqcup J} = \calC_I \oplus \calC_J$, where $\sqcup$ means that the sets do not intersect.
\end{itemize}
\end{theorem}

\begin{proof}
The code generating set for $\calC_I + \calC_J$ is the union of generating sets for $\calC_I$ and $\calC_J$, since 
$\bfy^{[I]}$ and $\bfy^{[J]}$ are generating sets respectively for $\calC_I$ and $\calC_J$ then $\bfy^{[I\cup J]}$ is a generating set for $\calC_I + \calC_J$. Hence $\calC_{I \cup J} = \calC_I + \calC_J$. Now the generating set $\bfy^{[I \cap J]}$ of $\calC_{I \cap J}$ is included in the generating set of $\calC_I$ and of $\calC_J$. Therefore $ \calC_{I \cap J} \subset \calC_I \cap \calC_J$.

Let us consider now $M$ such that $\bfy^{[M]}$ is formed with linearly independent vectors. It is immediate that for any $I \subset M$, a basis of $\calC_I$ is $\bfy^{[I]}$, therefore the dimension of $\calC_I$ is exactly equal to $|I|$. Let $\bfc \in  \calC_I \cap \calC_J$. We have 
\[
\bfc = \sum_{(u,j) \in I}{c_{u,j}\bfy^{[u,j]}} = \sum_{(u,j) \in I \setminus J}{c_{u,j}\bfy^{[u,j]}} +  \sum_{(u,j) \in I \cap J}{c_{u,j}\bfy^{[u,j]}}
\]
and similarly 
\[
\bfc = \sum_{(u,j) \in J}{c'_{u,j}\bfy^{[u,j]}} = \sum_{(u,j) \in J \setminus I}{c'_{u,j}\bfy^{[u,j]}} +  \sum_{(u,j) \in I \cap J}{c'_{u,j}\bfy^{[u,j]}}
\]
Since by hypotheses on $M$, the $\bfy^{[M]}$ are linearly independent, this implies the equality of the coefficients on this bases and thus that $c_{u,j} = 0$, for $(u,j) \in  I \setminus J$. Therefore, $\bfc \in \calC_{I \cap J}$.

The last item comes from the fact that is $I$ and $J$ do not intersect then $I\cap J = \emptyset$, thus $\calC_I \cap \calC_J = \{ 0 \}$.

\end{proof}

\subsection{A distinguishing attack in the general case}
\label{Sec:Dist}

We show that if $n,k,\lambda$ satisfy 
$k >  \frac{(\lambda-1)n}{\lambda} +1$, then one can distinguish the public-code from a random code in polynomial time.
First we prove the following theorem.

\begin{theorem}
\label{Theo:Dist} $ \dim_{\FF_{q^m}} \left( \cpubp + \cpubp^{[1]} + \dots + \cpubp^{[\lambda]} \right) \leq \lambda(n-k)+\lambda$ 
\end{theorem}

\begin{proof}
    Let $\calS_0  \stackrel{def}{=} \sum\limits_{i=0}^{\lambda-1} \cpubp^{[i]}$. We want to show that 
     $\dim\left(\calS_0 + \cpubp^{[\lambda]}\right) \le \lambda(n-k)+\lambda$. 
     For any $(\bfh,\vec{\beta}) \in \calS$ from the expression of  $\cpubp$ under the form (\ref{Eq:CpubIndices}) and from the CodeSet theorem  we obtain
     \[
\calS_0 =   \mathcal{C}_{S_0},\mbox{ where }S_0 =  \bigsqcup_{u=0}^{\lambda-1}{\{ u\}\times [u,n-k+u-1]} 
 \] 
    and 
     \[
 \cpubp^{[\lambda]} = \mathcal{C}_{\{\lambda\} \times [\lambda,n-k+\lambda-1]}
 \]

Let 
\[ I = \bigsqcup_{u=0}^{\lambda-1}{\{ u\}\times [\lambda-1,n-k-1]} = [0,\lambda-1] \times [\lambda-1,n-k-1] \] 
We have clearly  $I \subset S_0$, implying  $\calC_I \subset \calS_0$. 

By the hypotheses on $\calS$ the $(1,\beta_1,\ldots,\beta_{\lambda-1})$ are linearly independent over $\FF_q$. Thus, 
\[\det \begin{bmatrix}
  ~ & 1 & \beta_1 &\ldots & \beta_{\lambda-1} \\ 
  ~ & 1 & \beta_1^{[1]}  &\ldots &\beta_{\lambda-1}^{[1]}\\ 
  ~ & \vdots &  \vdots &\ddots &\vdots\\
  ~ & 1 & \beta_1^{{[\lambda-1]}} &\ldots &\beta_{\lambda-1}^{{[\lambda-1]}}
\end{bmatrix} \ne 0,
\] 
it implies that for any $j \in \lambda-1,\dots,n-k-1$, $\calC_{[0,\lambda-1] \times \{j\}} = \langle \bfh_i^{[j]}, 0 \le i \le \lambda-1 \rangle$. Hence,
\[
    \calC_I 
    = \left\langle \bfh_i^{[j]}, \begin{array}{l}
        0 \le i \le \lambda-1  \\
 		\lambda-1 \le j \le n-k-1 
 	  \end{array}
 \right\rangle
\]
In particular from the structure of $\bfy^{[u,j]}$ for any $J \subset * \times [\lambda-1,n-k-1]$, we have 
$\calC_{J} \subset \calC_I \subset \calS_0$.

Thus, 
$\mathcal{C}_{ \{\lambda\} \times [\lambda,n-k-1]} \subset \calS_0 \cap \cpubp^{[\lambda]}$. 
From its construction, $\cpubp^{[\lambda]}$ has dimension $n-k$. The vectors $\bfy^{[\lambda,j]},j \in [\lambda,n-k+\lambda-1] $ are linearly independent and from the CodeSet theorem,  $\mathcal{C}_{ \{\lambda\} \times [\lambda,n-k-1]}$ has dimension $n-k-\lambda$. Therefore, $\dim(\calS_0 \cap \cpubp^{[\lambda]}) \ge n-k-\lambda$. 
Conversely, 
 \[
\begin{array}{rcl}
\dim\left(\calS_0 + \cpubp^{[\lambda]}\right) &=&  \dim(\calS_0) + \dim(\cpubp^{[\lambda]}) - \dim(\calS_0 \cap \cpubp^{[\lambda]})
\\
& \le &\lambda(n-k) + (n-k) - (n-k-\lambda) = \lambda(n-k) + \lambda 
\end{array}
\]
\end{proof}
Now the distinguishing attack comes from this proposition 
\begin{proposition}[\cite{CC19} Proposition 2 ]
\label{Prop:RandCode}
If $\calC_{rand}$ is a random code of length $n$ and dimension $k$, then for a non-negative integer $a$ and a positive $\lambda<k$, we have
\[
    \mathbb{P}\left(\dim_{\FF_{q^m}}\left( \calC_{rand}+\calC_{rand}^{[1]}+\dots+\calC_{rand}^{[\lambda]}\right) \leq \min(n,(\lambda+1)k)-a\right) = O(q^{-ma}) .
\]
\end{proposition}

Now whenever  $k >   \frac{(\lambda-1)n}{\lambda} + 1$, the dimension of $\calC_{rand}+\calC_{rand}^{[1]}+\dots+\calC_{rand}^{[\lambda]}$ is very probably equal to $(\lambda+1)(n-k)$ whereas the dimension of $ \cpubp + \cpubp^{[1]} + \dots + \cpubp^{[\lambda]}$ is probably equal to $\lambda(n-k+1)$ (since $\lambda(n-k+1) <n$) , which is strictly less than $(\lambda+1)(n-k)$.

\subsection{Reconstructing attack}
We suppose that the public code has rate larger than $(\lambda-1)/\lambda$, so that the distinguisher introduced in Section \ref{Sec:Dist} works on it.


Although the attack we describe should work heuristically, to have rigorous proofs of work we  need the following assumptions, which are not very contraining 

\begin{enumerate}[label=(\arabic*)]
	
    \item \label{asp:1}There  exists an element $(\bfh,\vec{\beta}) \in \mathcal{S}$ such that 
    $\forall i_1,\dots,i_{\lambda} \in\{1,\ldots,n-k-1\}$ distinct.
    \[
    \det\begin{bmatrix}
~ & 1 & \beta_1^{[i_1]} &\beta_2^{[i_1]} &\ldots & \beta_{\lambda-1}^{[i_1]} \\ 
  ~ & 1 & \beta_1^{[i_2]} & \beta_2^{[i_2]} &\ldots &\beta_{\lambda-1}^{[i_2]}\\ 
  ~ & \vdots & \vdots & \vdots &\ddots &\vdots\\
  ~ & 1 & \beta_1^{{[i_{\lambda]}}} &\beta_2^{{[i_{\lambda}]}} &\ldots &\beta_{\lambda-1}^{{[i_{\lambda}]}}
\end{bmatrix} \ne 0, 
\]
    \item \label{asp:2}$\dim_{\FF_{q^m}}\cpubp + \cpubp^{[1]} + \cpubp^{[2]} + \dots + \cpubp^{[\lambda]} = \lambda(n-k)+\lambda$
    \item \label{asp:3} There is no $\overline{\bfA} \in \textbf{PGL}(\lambda,\FF_q) \backslash \overline{\bfI_{\lambda}}$ and $\bfA=  (a_{ij})_{i,j=1}^{\lambda} $ that satisfies
    
    \[
        \beta_j = \dfrac{a_{0,j} +\sum\limits_{i=1}^{\lambda-1}a_{i,j}\beta_i}{a_{0,0} + \sum\limits_{i=1}^{\lambda-1}a_{i,0}\beta_i}, ~\forall j = 1,\dots, \lambda -1
    \]
    
\end{enumerate}
The first step of the attack is dedicated to finding one dimensional vector-spaces
$\calA_i$ for $i=1,\ldots,n-k-1$, such that any element  $(\bfh,\vec{\beta}) \in \mathcal{S}$ satisfies: 
\[
 \calA_i = \left\langle \bfh_0 + \sum_{j=1}^{\lambda-1} \beta_j^{[-i]}\bfh_j \right\rangle
 \]
 
From the $\calA_i$'s, one obtains a system of $\lambda-1$ multivariate polynomials which are of degree $q^{\lambda+1}-q^i$ for $i=1,\dots,\lambda-1$ satisfied by all the vectors $\vec{\beta}$ such that   $(\bfh,\vec{\beta}) \in \mathcal{S}$. Under the above assumptions, we can also prove the stabilization of the set of solution $\calS$ under the action of \textbf{PGL}$(\lambda,\FF_q)$ in the end of 3.2.2.

The complexities of the steps (by operations over $\FF_{q^m}$):

\begin{itemize}
    \item Step 1. It costs $O(n^3\log q)$ operations for computing $\cpubp^{[i]}$ and  $O(n^{\omega+1})$ for taking the intersection.
    \item Step 2. The principal complexity of this is finding the roots of the system of polynomials. In case of $\lambda=3$, it can be done in polynomial time where the complexity is $\tilde{O}(\tilde{d}^2n \log q)$ for $\tilde{d} = (q^4-q)(q^4-q^2)$.
    
    Once such a root is found, the remaining of step 2 needs a finite number of linear systems solving which costs $O(n^{\omega})$.
\end{itemize}

\subsubsection{First step: Recovering one-dimensional vector spaces}

We now suppose that the three assumptions in section \ref{Sec:Solutions} are true we have the following theorem:

\begin{theorem}\hfill 
\begin{algorithm}[h!] 
\caption{Recovering $1$-dimensional vector spaces}
\label{algo:recover}
\KwData{$\cpubp$, $\lambda \le \dfrac{n}{n-k+1}$}
\KwResult{$\calA_i$ for $i = 0,\dots,n-k-1$}

$\calS_0 \leftarrow \cpubp^{[0]} + \cpubp^{[1]} + \cdots + \cpubp^{[\lambda-1]} $\\

$\calA  \leftarrow \left(\bigcap\limits_{i=0}^{d}\calS_0^{[i]}\right)^{[-d]}$ \\

$\calD_{\lambda-1}  \leftarrow \calA^{[\lambda-2]} \cap \cpubp^{[\lambda-1-d]}$ and $\calB_0 \leftarrow \calA+\calD_{\lambda-1}^{[1-\lambda]}$

$\calD_0 \leftarrow \calB_0 \cap \cpubp^{[-1]}$

\For{$\ell \in 1,\dots,\lambda-2$}{
  $\calB_\ell \leftarrow  \calA +  \sum\limits_{j=0}^{\ell-1}\calD_j^{[\ell-j]}$\;
    $\calD_\ell \leftarrow \calB_{\ell} \cap \cpubp^{[-1]}$
 }
 
$\calH \leftarrow \sum\limits_{j=0}^{\lambda-1} \calC_j^{[2-j-\lambda]}$\\

\For{$i \in 0,\dots,n-k-1$}{
  Return $\calA_i \leftarrow \calH \cap \cpubp^{[-i]} $
 }
\end{algorithm}
  Let $d := n-k-\lambda +1$. Under Assumptions \ref{asp:1}, \ref{asp:2}, \ref{asp:3}, Algorithm (\ref{algo:recover}) returns the $1$-dimensional vector spaces 
\[
\calA_i = \left\langle \bfh_0 + \sum\limits_{j=1}^{\lambda-1}\beta_j^{[-i]} \bfh_j \right\rangle, ~i = 0,\dots,n-k-1
\] 
for any $(\bfh,\vec{\beta}) \in \mathcal{S}$.

\end{theorem}

\begin{proof} 
 For the proof we will thus make intensive use of the CodeSet theorem. First from assumption (2) and theorem \ref{Theo:Dist}, the set 
 \[
 M =  \bigsqcup_{u=0}^{\lambda}{\{ u\}\times [u,n-k+u-1]} = S_0 \sqcup \{ \lambda \} \times [\lambda,n-k+\lambda-1] 
 \]
with cardinality $\lambda(n-k)+ \lambda$, is such that $\bfy^{[M]}$ is formed of linearly independent vectors and $\calC_M = \sum\limits_{i=0}^{\lambda}\cpubp^{[i]}$. This point is very important since this is the crucial point of the proof of the theorem. 

\noindent\textbf{Line 1.} 
 From theorem \ref{Theo:Dist} we have
        \[
 \calS_0 =   \mathcal{C}_{S_0},\mbox{ where }S_0 =   \bigsqcup_{u=0}^{\lambda-1}{\{ u\}\times [u,n-k+u-1]}
 \] 
   
 We can write $S_0$ under the form 
 \[
 S_0 =  \underbracket[.5pt]{\left( \bigsqcup_{u=0}^{\lambda-2}{\{ u\}\times [u,\lambda-2]} \right)}_{I_1} \sqcup 
  \underbracket[.5pt]{[0,\lambda-1]\times [\lambda-1,n-k-1]}_{I_2} \sqcup 
 \underbracket[.5pt]{\left( \bigsqcup_{u=0}^{\lambda-1}{\{ u\}\times [n-k,n-k+u-1]} \right)}_{I_3} 
\]
Let $I_4 = I_3 \sqcup \{\lambda\} \times [n-k,n-k+\lambda-1]$. With these notations, we have

$S_0 = I_1 \sqcup I_2 \sqcup I_3 \subset I_1 \sqcup I_2 \sqcup I_4 = M$. 
Since their cardinalities satisfy 

$|I_1| = \frac{\lambda(\lambda-1)}{2}$, $|I_2| = \lambda(n-k-\lambda+1)$, $|I_3| = \frac{\lambda(\lambda-1)}{2} $ and $|I_4| = \lambda $,

from theorem \ref{Theo:Dist} the dimension of $\calS_0$ and $\calC_M$ is exactly $\lambda (n-k)$ and $\lambda (n-k) + \lambda$ respectively, and additionally under the CodeSet theorem,
\begin{align*}
    \calS_0 &= \calC_{I_1} \oplus \calC_{I_2} \oplus \calC_{I_3}\\
    \calC_M &= \calC_{I_1} \oplus \calC_{I_2} \oplus \calC_{I_4}
\end{align*}

The set $I_2$ corresponds to the set denoted by $I$ in the proof of theorem \ref{Theo:Dist}. We have 
\[
    \calC_{I_2} 
    = \left\langle \bfh_i^{[j]}, \begin{array}{l}
        0 \le i \le \lambda-1  \\
 		\lambda-1 \le j \le n-k-1 
 	  \end{array}
 \right\rangle
\]

This property gives us the flexibility for the modification of the set $M$ to obtain several sets of indexes $M'$ such that $\bfy^{[M']}$ is formed of linearly independent vectors. It can be done by the replacement of the set $[0,\lambda-1]$ by the set $A_j$  of $\lambda$ elements corresponding to any $j$. We can see it precisely as the following lemma: 

\begin{lemma}
    For every set $I_2' = \bigsqcup\limits_{j=\lambda-1}^{n-k-1}A_j \times \{j\}$ where $|A_j| = \lambda$, then $M' = I_1 \sqcup I_2' \sqcup I_4$ satisfies
\begin{itemize}
    \item $\calC_M = \calC_{M'}$.
    \item $\bfy^{[M']}$ is formed of linearly independent vectors.
\end{itemize}
\end{lemma}

\begin{proof}~
$\calC_{I_2'} = \left\langle \bfh_i^{[j]}, \begin{array}{l}
        0 \le i \le \lambda-1  \\
 		\lambda-1 \le j \le n-k-1 
 	  \end{array}
 \right\rangle = \calC_{I_2}$ (Assumption~\ref{asp:1}). Moreover, $|I_2'| = |I_2| = \lambda(n-k-\lambda+1) $. Hence $\calC_{M'} = \calC_{I_1} \oplus \calC_{I_2'} \oplus \calC_{I_4} = \calC_M$ and $\bfy^{[M']}$ is formed of linearly independent vectors.
\end{proof}

Through this section, to apply the CodeSet theorem, in the beginning of each step, we will define its set of indexes $M'$ such that $\bfy^{[M']}$ are linearly independent vectors and it contains the set of indexes of subspace that we want to compute the intersection. To be convenient, we will use some images where the red dot ${\color{red} \bullet}$ are indexes of some transformation of $I_1$, the blue square ${\color{blue} \sq}$ are indexes of some transformation of $I_4$, the green ${\color{green} \times}$ are indexes of some transformation of $I_2$ and the black diamond $\smallblackdiamond$ are indexes of $\cpubp$. On the other hand, the integer points which are inside the blue figures are indexes of linearly independent vectors. The left triangular covers all the points of the set of indexes $I_1$, the right triangular covers all the points of the set of indexes $I_4$ and the rectangular covers all the points of the set of indexes which is flexible modification of $I_2$.

To be convenient, for $I \subset \ZZ \times \ZZ$ and $a \in \ZZ$, we denote $I + a = \{(u+a,j+a), ~(u,j) \in I\}$. In the figures, we can consider $I+a$ as the translation of $I$ by the vector $(a,a)$. For example, in the figure \ref{fig:fig1}, the set of red points shows $I_1$ in the first image and $I_1+d+1$ in the second ones.

\noindent\textbf{Line 2.} We show that
$
\calA = \calC_{I_1 \sqcup I_3+(\lambda-(n-k)-1)}
$.
\begin{lemma}
\label{Lemma:Inter}
Let $\calS_i\stackrel{def}{=} \cpubp^{[i]} + \cpubp^{[i+1]} + \cdots + \cpubp^{[i+\lambda-1]}$. For any set $*$ of $\lambda$ distinct integers modulo $m$ we have 
   \[
            \forall 0  \le d \le n-k-\lambda+1, ~ \bigcap\limits_{i=0}^{d}\calS_i  =  \calC_{(I_1+d) \sqcup * \times [\lambda-1+d,n-k-1] \sqcup I_3} 
    \]
\end{lemma}

\begin{proof}
    We prove the theorem by induction. This lemma is true for $d = 0$. We suppose that it is true until $0 \le d   \le n-k-\lambda$, then we need to prove that it must be true for $d+1$.
    
    Indeed,

\begin{equation*}
    \begin{aligned}
     \bigcap_{i=0}^{d+1} \calS_i  
    &=&  &  \calS_0 \cap \left(  \bigcap \limits_{i=0}^{d} \calS_i  \right)^{[1]}  \\
    &=&  &\calC_{S_0}  \cap  
    	        ~  \calC_{(I_1+d+1) \cup * \times [\lambda+d,n-k] \cup (I_3+1)}
    \end{aligned}
\end{equation*}

\begin{figure}[h]
    \centering
    \includegraphics[scale=0.3]{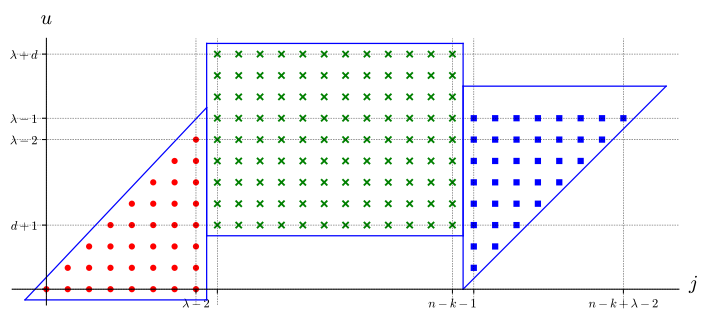}
    \includegraphics[scale=0.3]{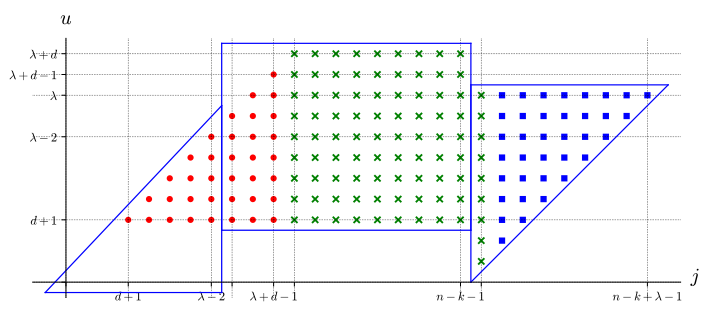}
    \caption{Points of $\calS_0$ (above) and $(I_1+d+1) \cup * \times [\lambda+d,n-k] \cup (I_3+1)$}
    \label{fig:fig1}
\end{figure}

Let $M_1 = I_1 \sqcup I_4 \sqcup [d+1,\lambda+d] \times [\lambda-1,n-k-1]$

Concerning $S_0 \subset M_1$ and $J \stackrel{def}{=}(I_1+d+1) \cup * \times [\lambda+d,n-k] \cup (I_3+1) \subset M_1$. We can apply the CodeSet theorem  
\[
\calC_{S_0} \cap \calC_{J} = \calC_{S_0 \cap J}
\]
And by a slightly fastidious computation on the sets intersections, we see that
\[
S_0 \cap J = (I_1+d+1) \cup * \times [\lambda+d,n-k-1] \cup I_3
\]
It is not very difficult to check that the sets do not intersect which gives the result.

\end{proof}

In the rest of the proof we will suppose that  $d = n-k-\lambda+1$. If we instantiate the lemma with $d$ and elevate to the power $[-d]$ we obtain the following corollary:
\begin{corollary}
$\calA = \calC_{I_1 \sqcup I_3-d}$
\end{corollary}

\begin{proof}
We have $\bigcap\limits_{i=0}^{d}\calS_i  =  \calC_{(I_1+d) \sqcup  I_3}$, with $I_1$ and $I_3$ subsets of $M$. Thus, from CodeSet theorem we have 
\[
\calA^{[d]} =  \calC_{I_1+d} \oplus \calC_{I_3}
\] 
Implying that $\calA =  \calC_{I_1} \oplus \calC_{I_3-d}$. 


\end{proof}

\noindent\textbf{Line 3.}
Since $ \cpubp = \calC_{\{0\} \times [0,n-k-1]}$, we have
\[
    \calA^{[d - 1]} \cap \cpubp = \calC_{(I_1+(d-1)) \sqcup (I_3 -1)} \cap \calC_{\{0\} \times [0,n-k-1]}
\]
with
\begin{align*}
    I_1+(d- 1) &= \{(u,j): d - 1 \le u \le j \le n-k-2\}\\
    I_3 - 1 &= \{(u,j): 0 \le u \le \lambda-2, ~n-k-1 \le j \le n-k+u-1 \}
\end{align*}

\begin{figure}[h]
    \centering
    \includegraphics[scale=0.3]{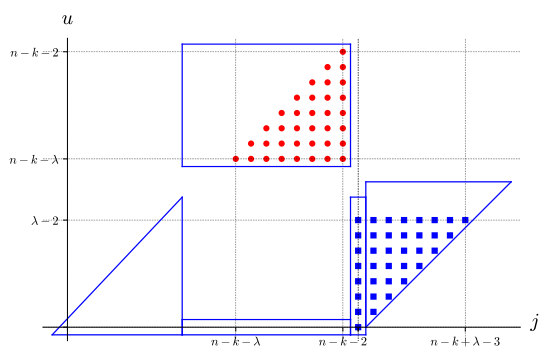}
    \includegraphics[scale=0.3]{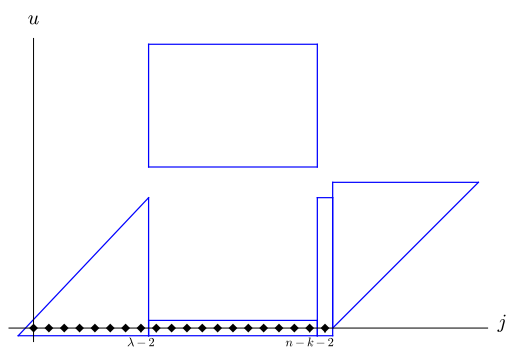}
    \caption{Points of $(I_1+(d-1)) \sqcup (I_3 -1)$ and $\{0\} \times [0,n-k-1]$}
    \label{fig:fig2}
\end{figure}

Let 
\[
    I_2' = \left((\{0\} \sqcup [n-k-\lambda,n-k-2]) \times [\lambda-1,n-k-2]\right) \sqcup [0,\lambda-1] \times \{n-k-1\} 
\]
and $M_2 = I_1 \sqcup I_4 \sqcup I_2'$

We can prove that $I_1+(d- 1),  I_3 - 1$ and $\{0\} \times [0,n-k-1]$ are all in $M_2$. Now since $\lambda \le (n-k)/2$, we have $d \ge \lambda-1$, it implies that 
\[
    \left( (I_1+(d-1)  \sqcup (I_3 - 1 ) \right) \cap \{0\} \times [0,n-k-1] = \{ (0,n-k-1) \}
\]
Since the intersecting sets are all subsets of the set $M_2$, we can apply the 
CodeSet theorem, and we obtain
\[
    \calA^{[d-1]} \cap \cpubp = \calC_{\{(0,n-k-1)\}}  
\]
which by elevating to the power $[\lambda-d-1]$  gives 
\[
\calD_{\lambda-1} = \calA^{[\lambda-2]} \cap \cpubp^{[\lambda-d-1]} = \calC_{\{(\lambda-d-1,2\lambda-3)\}}
\]
Note that elevating to the power $[m]$ the scalars corresponds to the identity operator, we also have 
 $\calD_{\lambda-1} = \calC_{\{(m+\lambda-d-1,2\lambda-3)\}}$. This will be of use in the proof of the algorithm. 
Now since  $\calB_0 = \calA+\calD_{\lambda-1}^{[1-\lambda]}$, we deduce that 
\[
    \calB_0 = \calC_{I_1 \sqcup I_3 - d} + \calC_{\{(-d,\lambda-2)\}}
\]

\textbf{Line 4.} Compute $\calD_0 = \calB_0 \cap \cpubp^{[-1]}$

We compute
\begin{align*}
    \calD_0^{[1]} &= \calB_0^{[1]} \cap \cpubp\\
    &= \left(\calC_{(I_1+1) \sqcup I_3 - (d-1)} + \calC_{\{(-d+1,\lambda-1)\}}\right) \cap \calC_{\{0\} \times [0,(n-k)-1]}\\
    &= \left(\calC_{(I_1+1) \backslash [1,\lambda-1] \times \{\lambda-1\}} \oplus \calC_{* \times \{\lambda-1\}} \oplus \calC_{I_3 - (d-1)}\right) \cap \calC_{\{0\} \times [0,(n-k)-1]}
\end{align*}
where * is instantiated for the set $[0,\lambda-2] \sqcup \{-d+1\}$ which contains $\lambda$ different integers.

\begin{figure}[h]
    \centering \includegraphics[scale=0.3]{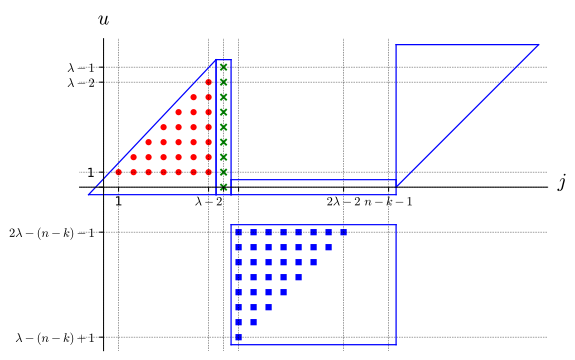}
    \includegraphics[scale=0.3]{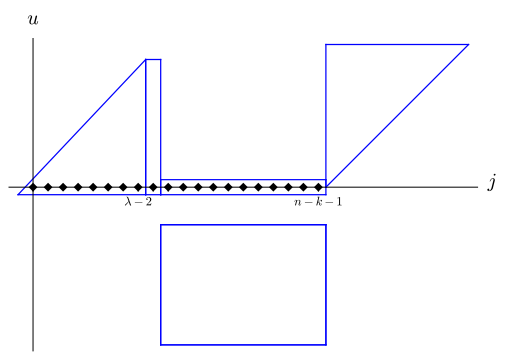}
    \caption{Points of $\left((I_1+1) \backslash [1,\lambda-1] \times \{\lambda-1\}\right) \sqcup * \times \{\lambda-1\} \sqcup I_3 - (d-1)$ and $\{0\} \times [0,(n-k)-1]$}
    \label{fig:fig3}
\end{figure}

Let $I_2' = ([0,\lambda-1] \times \{\lambda-1\}) \sqcup ([2-d,\lambda-d] \times [\lambda,n-k-1])$ and $M_3 = I_1 \sqcup I_4 \sqcup I_2'$

Now $(I_1+1) \sqcup I_3 - (d-1) \sqcup (-d+1,\lambda-1) \subset M_3$. Moreover,

\[
    \left((I_1+1) \sqcup I_3 - (d-1) \right) \cap \{0\} \times [0,(n-k)-1] = \emptyset
\]
and $* \times \{\lambda-1\} \cap  \{0\} \times [0,n-k-1] = \{(0,\lambda-1)\} $

Hence, by the CodeSet theorem,
\begin{align*}
    \calB_0^{[1]} \cap \cpubp &= \calC_{\{(0,\lambda-1)\}}\\
    \calB_0 \cap \cpubp^{[-1]} &= \calC_{\{(m-1,\lambda-2)\}}
\end{align*}

\textbf{Line 5,6,7} For $1 \le i \le \lambda-2$, we compute
\begin{align*}
    \calB_i &= \calA +  \sum\limits_{j=0}^{i-1}\calC_j^{[i-j]}\\
    \calD_i &= \calB_{i} \cap \cpubp^{[-1]} = \calC_{\{(m-1,\lambda+i-2)\}}
\end{align*}

We prove by induction, suppose that $\calD_i = \calC_{\{(m-1,\lambda+i-2)\}}$ for all $1 \le i \le \ell \le \lambda-2$. we prove this for $i = \ell+1$
\begin{align*}
    \calB_{\ell+1} &= \calA +  \sum\limits_{j=0}^\ell\calC_j^{[\ell+1-j]}\\
    \calD_{\ell+1} &= \calB_{\ell+1} \cap \cpubp^{[-1]} = \calC_{\{(m-1,\lambda+\ell-1)\}}
\end{align*}

Indeed,

\begin{align*}
    \calB_{\ell+1} &= \calA + \sum\limits_{j=0}^{\ell}\left(\calC_{\{(m-1,\lambda+j-2)\}}\right)^{[\ell+1-j]}\\
    &= \calA + \sum\limits_{j=0}^\ell \calC_{\{(\ell-j,\lambda+\ell-1)\}}\\
    &= \calC_{I_1 \sqcup I_3-d} + \calC_{[0,\ell] \times \{\lambda+\ell-1\}}\\
    &= \calC_{I_1 \sqcup (I_3-d \backslash [\ell-d+1,\lambda-d-1]\times\{\lambda+\ell-1\})} + \calC_{* \times \{\lambda+\ell-1\}}
\end{align*}
where * is instantiated for the set $[\ell-d+1,\lambda-d-1] \sqcup [0,d]$ of $\lambda$ distinct integers.

\begin{figure}[h]
    \centering     \includegraphics[scale=0.3]{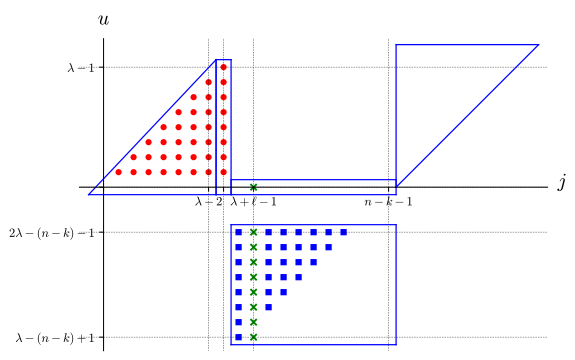}
    \includegraphics[scale=0.3]{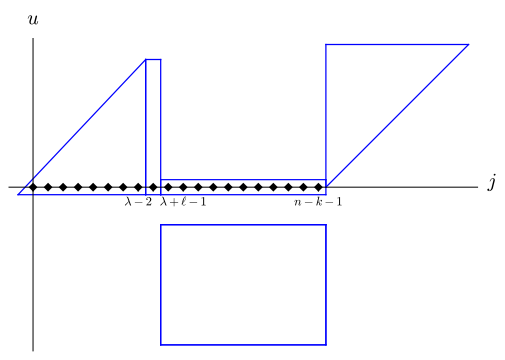}
    \caption{Points of $I_1 \sqcup (I_3-d \backslash [\ell-d+1,\lambda-d-1]\times\{\lambda+\ell-1\}) \sqcup * \times \{\lambda+\ell-1\}$ and $\{0\} \times [0,(n-k)-1]$}
    \label{fig:fig4}
\end{figure}

We compute $\calD_{\ell+1}^{[1]} = \calB_{\ell+1}^{[1]} \cap \cpubp$. Let 
\[
I_2' = ([0,\lambda-1] \times \{\lambda-1\}) \sqcup (\{0\} \sqcup [2-d,\lambda-d] ) \times [\lambda,n-k-1]
\]
 and  $M_4 = I_1 \sqcup I_4 \sqcup I_2'$. Then, $I_1+1, I_3-d+1$ and $* \times \{\lambda+\ell+1\} \subset M_4$.

Moreover,

\begin{align*}
    I_3-d+1 \cap \{0\} \times [0,n-k-1] &= \varnothing\\
    I_1+1 \cap \{0\} \times [0,n-k-1] &= \varnothing
\end{align*}

Hence, by CodeSet theorem,
\begin{align*}
    \calB_{\ell+1}^{[1]} \cap \cpubp &= \calC_{* \times [\lambda+\ell]\cap \{0\} \times [0,n-k-1]}\\
    &= \calC_{\{(0,\lambda+\ell)\}}
\end{align*}

Therefore, $\calD_{\ell+1} = \calB_{\ell+1} \cap \cpubp^{[-1]} = \calC_{\{(m-1,\lambda+\ell-1)\}} $

\textbf{Line 8.}  Compute $\calH =  \sum\limits_{j=0}^{\lambda-1} \calC_j^{[2-j-\lambda]}$.

We consider the sum of subspace:
\begin{align*}
    & \sum\limits_{j=0}^{\lambda-1} \calC_j^{[2-j-\lambda]} = \sum\limits_{j=0}^{\lambda-2} \calC_{\{(m+1-\lambda-j,0)\}} +  \calC_{\{(m+1-(n-k),0)\}} \\
    &= \calC_{([m+3-2\lambda,m+1-\lambda] \sqcup \{m+1-(n-k)\}) \times \{0\} } = \calC_{* \times \{0\}} =: \calH
\end{align*}
where * is instantiated for the set $[m+2-2\lambda,m+1-\lambda] \sqcup \{m+1-(n-k)\}$ of $\lambda$ distinct integers.

\textbf{Line 9-10} Next, for any $i \in \{0,\dots n-k-1\}$, one can compute
\[
    \cpubp^{[-i]} \cap \calH = \langle \bfh_0 + \sum\limits_{i=1}^{\lambda-1}\beta_i^{[-i]} \bfh_i\rangle 
\]

\begin{itemize}
    \item For $\lambda-1 \le i \le n-k-1$, $* \times \{i\} \subset M$
    \[
    \{(0,i)\} = (\{0\} \times [0,n-k-1]) \cap (* \times \{i\})
    \]
    Hence, by the CodeSet theorem, $\calC_{(0,i)} =  \cpubp \cap \calH^{[i]}$
    \item For  $0 \le i \le \lambda-2$, $\cpubp^{[\lambda-1]}  = \calC_{\{\lambda-1\} \times [\lambda-1,n-k+\lambda-2] }$ and $\calH^{[\lambda+i-1]} = \calC_{*\times\{\lambda+i-1\}}$. Moreover, $*\times\{\lambda+i-1\} \subset *\times [\lambda-1,2\lambda-3] \subset M$ and
    \[
    \{(\lambda-1,\lambda+i-1)\} = (\{\lambda-1\} \times [\lambda-1,n-k+\lambda-2]) \cap  (*\times\{\lambda+i-1\}))
    \]
    Hence, by the CodeSet theorem, 
    \begin{align*}
        \cpubp^{[\lambda-1]} \cap \calH^{[\lambda+i-1]} &= \calC_{\{(\lambda-1,\lambda+i-1)\}}  \\
        \cpubp \cap \calH^{[i]} &= \calC_{\{(0,i)\}} 
    \end{align*}

    Therefore, for any $i \in \{0,\dots n-k-1\}$, one can compute
\[
    \cpubp^{[-i]} \cap \calH = \calC_{\{(-i,0)\}} =  \langle \bfh_0 + \sum\limits_{i=1}^{\lambda-1}\beta_i^{[-i]} \bfh_i\rangle 
\]
\end{itemize}



\end{proof}

Note that this specialization of one element of $\mathcal{S}$ should be true for any element in $\mathcal{S}$. Indeed, for 2 elements $(\bfh',\beta')$ and $(\bfh,\beta)$, if there exists $\bfA \in GL(\lambda,\FF_q)$ such that $(\bfh',\beta') = \bfA \cdot (\bfh,\beta)$, then $\langle \bfh_0 + \sum\limits_{j=1}^{\lambda-1}\beta_j^{[-i]} \bfh_j\rangle = \langle \bfh_0 + \sum\limits_{j=1}^{\lambda-1}{\beta'}_j^{[-i]} {\bfh'}_j\rangle$ since ${\bfh'}_0 + \sum\limits_{j=1}^{\lambda-1}{\beta'}_j^{[-\ell]}{\bfh'}_j = \left(a_{0,0} + \sum\limits_{i=1}^{\lambda-1}a_{i,0}\beta_i\right)^{1-[-\ell]}\left(\bfh_0 + \sum\limits_{i=1}^{\lambda-1}\beta_i^{[-\ell]}\bfh_i \right)$

\subsubsection{Second step: Recovering the vector space}\hfill

\noindent From step 1, we recovered the $1$-dimensional vector-spaces 
\[
\forall i=0,\ldots,n-k-1, ~ \calA_i = \left\langle \bfh_0 + \sum\limits_{j=1}^{\lambda-1}\beta_j^{[-i]} \bfh_j \right\rangle
\]
The vector spaces  $\calA_i$ do not depend on $(\bfh, \mathbf{\beta})\in \calS$.
We introduce the following lemma.

\begin{lemma}
    \label{Lemma:Sum}
    For any $\bfu_0 \in \calA_0$, and for any set $\calI = \{ i_1,\ldots,i_{\lambda} \} \subset \{1,\dots,n-k-1\}$  of $\lambda$
    distinct elements, there exists a unique $\lambda$-tuple  
    $\bfu_\calI \stackrel{def}{=}(\bfu_{i_1},\bfu_{i_2},\ldots,\bfu_{i_\lambda}) \in \bigtimes\limits_{j=1}^{\lambda}{\calA_{i_j}}$ such that 
    \[ 
    \sum\limits_{i_j \in \calI}\bfu_{i_j} = \bfu_0
    \]
\end{lemma}

\begin{proof}
    We observe that, from assumption (1) we have 
    \[
            \calA_{i_1} \oplus \cdots \oplus \calA_{i_\lambda} = \langle \bfh_0,\dots,\bfh_{\lambda-1} \rangle.
    \]
    Since $\calA_0 \subset \langle \bfh_0,\dots,\bfh_{\lambda-1} \rangle$, this completes the proof.
\end{proof}

We denote $k_{i_\ell} \in \FF_{q^m}$ such that $\forall i_\ell \in \calI, \bfu_{i_\ell} = k_{i_\ell}\left(\bfh_0 + \sum\limits_{j=1}^{\lambda-1} \beta_j^{[-i_\ell]}\bfh_j\right)$. A vector $\bfu_0 \in \calA_0$ can be written under the form 
\[
\bfu_0 = \alpha_{\bfh,\mathbf{\beta}}(\bfh_0 + \sum_{j=1}^\lambda{\beta_j \bfh_j})
\]
From the structure of the solution space $\calS$, there exists an $(\bfh,\mathbf{\beta}) \in \calS$ such that $\alpha_{\bfh,\mathbf{\beta}}=1$. It means that we can fix $
{\color{blue}\bfu_0} := \bfh_0 + \sum_{j=1}^\lambda{\beta_j \bfh_j}$ as a known vector. For this element and for any $\calI = \{i_1,\ldots,i_{\lambda} \}$ from Lemma \ref{Lemma:Sum}  we have

\begin{align*}
    \sum \limits_{i_\ell \in \calI } k_{i_\ell}^\calI \left(\bfh_0 + \sum\limits_{j=1}^{\lambda-1} \beta_j^{[-i_\ell]}\bfh_j\right) 
    &= \left(\sum \limits_{i_\ell \in \calI }k_{i_\ell}^\calI  \right)\bfh_0 + \sum\limits_{j=1}^{\lambda-1}\left(\sum \limits_{i_\ell \in \calI}k_{i_\ell}^\calI  \beta_j^{[-i_\ell]} \right)\bfh_j\\
    &= \bfh_0 + \sum\limits_{j=1}^{\lambda-1}\beta_j\bfh_j
\end{align*}
Since the $\bfh_j$ are linearly independent we obtain the following system
\[
    (k_{i_1}^\calI ,k_{i_2}^\calI ,\dots,k_{i_\lambda}^\calI )\begin{bmatrix}
  ~ & 1 & \beta_1^{[-i_1]} &\beta_2^{[-i_1]} &\cdots &\beta_{\lambda-1}^{[-i_1]}\\ 
  ~ & 1 & \beta_1^{[-i_2]} & \beta_2^{[-i_2]} &\cdots &\beta_{\lambda-1}^{[-i_2]}\\ 
  ~ & \vdots & \vdots &\vdots&\ddots &\vdots\\
  ~ & 1 & \beta_1^{[-i_{\lambda}]}&\beta_2^{[-i_{\lambda}]} &\cdots &\beta_{\lambda-1}^{[-i_{\lambda}]}
\end{bmatrix} = (1,\beta_1,\beta_2,\dots,\beta_{\lambda-1})
\]
in the unknowns $k_i^\calI $ and $\beta_i$. From assumption ~\ref{asp:1}, knowing the $\beta_i$'s, the solution is unique. 
To solve the system,  let us consider the associate matrix 
\[
    Mat^\calI(\vec{X}) := \begin{bmatrix}
  ~ & 1 & X_1^{[i_1]} &X_2^{[i_1]} &\cdots &X_{\lambda-1}^{[i_1]}\\ 
  ~ & 1 & X_1^{[i_2]} & X_2^{[i_2]} &\cdots &X_{\lambda-1}^{[i_2]}\\ 
  ~ & \vdots & \vdots &\vdots&\ddots &\vdots\\
  ~ & 1 & X_1^{[i_{\lambda}]}&X_2^{[i_{\lambda}]} &\cdots &X_{\lambda-1}^{[i_{\lambda}]}
\end{bmatrix}
\]
where  $\vec{X} = (X_1,X_2,\dots,X_{\lambda-1})$ is formed with the unknowns. We define the multivariate polynomial 
\[
f^\calI(\vec{X})\stackrel{def}{=}  \det(Mat^\calI(\vec{X}))
\]
Since  $f^\calI \in \FF_{q}[\vec{X}]$ we have
\begin{lemma}
\label{Lemma:PolyMult}
$f^{\calI}$ has degree $\sum_{j \in \calI}{[j]}$ and 
for all $u \in \mathbb{Z}$, 
    $f^{\calI + u }(\vec{X}) = f^{\calI}(\vec{X})^{[u]}$.
\end{lemma}

By Cramer's rule, for any $j=1,\ldots,\lambda$ we have
\begin{equation}
\label{Eq:Cramer}
    k_{i_j}^\calI = \dfrac{f^{-(\calI \setminus \{ i_j\} ) \cup \{ 0 \}}(\vec{\beta})}{ f^{-\calI}(\vec{\beta})},
\end{equation}
where $\vec{\beta} = (\beta_1,\ldots,\beta_\lambda)$. Let us define 
$\calJ_s = (\{1,\ldots, \lambda+1\}) \setminus \{ s+1 \} $, for all $s=1,\ldots,\lambda$. From (\ref{Eq:Cramer}), we have 
\[
 \forall s \in \{1,\ldots,\lambda  \},~ k_1^{\calJ_s} = \dfrac{f^{-(\calJ_s \setminus \{ 1\} ) \cup \{ 0 \}}(\vec{\beta})}{ f^{-\calJ_s}(\vec{\beta})}
\]
By elevating the equation to the power $[\lambda+1]$, from Lemma \ref{Lemma:PolyMult} we have
\[
\forall s \in \{1,\ldots,\lambda\},~ (k_1^{\calJ_s})^{[\lambda+1]} = \dfrac{f^{(\lambda+1)-(\calJ_s \setminus \{ 1\} ) \cup \{ 0 \}}(\vec{\beta})}{ f^{(\lambda+1)-\calJ_s}(\vec{\beta})}
\]

Now since we know only the vector space $\calA_1$ and not the exact vectors  $ \bfh_0 + \sum\limits_{j=1}^{\lambda-1}\beta_j^{[-1]} \bfh_j$, we do not know $k_1^{\calJ_s}$. However, we can compute the quantity 
$k_1^{\calJ_\lambda}/k_1^{\calJ_s}$ for $s \in \{1,\ldots, \lambda-1 \}$ thank to Algorithm \ref{algo:as} and Lemma \ref{Lemma:Sum}.
\bigskip

\begin{algorithm}[H]
\label{algo:as}
\KwData{$\{\calA_i\}_{i=1}^{n-k-1}$, $\{\calJ_s\}_{s=1}^{\lambda}$ and the vector $\bfu_0 \in \mathcal{A}_0$}
\KwResult{$\alpha_s = k_1^{\calJ_\lambda}/k_1^{\calJ_s}$ for $s \in \{1,\ldots, \lambda-1 \}$}

For $i = 1,\dots,n-k-1$, fix $\bfu_i$ arbitrarily in $\calA_i$ \\

For $s = 1,\dots,\lambda$, find  $a_j^{\calJ_s}$ such that of $\sum\limits_{j \in \calJ_s}a_j^{\calJ_s}\bfu_j = \bfu_0$\\

Return  $\dfrac{a_1^{\calJ_\lambda}}{a_1^{\calJ_s}}$, for $s = 1,\dots,\lambda-1$

\caption{Determining quotient $k_1^{\calJ_\lambda}/k_1^{\calJ_s}$}
\end{algorithm}

\bigskip
Now let us define by $\alpha_s = (k_1^{\calJ_\lambda}/k_1^{\calJ_s})^{[\lambda+1]}$, for $s = 1,\dots,\lambda-1$. To simplify notations, we also define 
\[
\forall s \in \{1,\ldots, \lambda\}
~ \left\{ \begin{array}{l}
\mathcal{L}_s =  (\lambda+1)-(\calJ_s \setminus \{ 1\}  \cup \{ 0 \})\\
\mathcal{M}_s = (\lambda+1)-\calJ_s
\end{array}
\right.
\]

We obtain the set of equations

\[
 \forall s \in \{1,\ldots, \lambda-1\}, \quad f^{\calL_\lambda}(\vec{\beta})f^{\calM_s}(\vec{\beta}) - \alpha_s f^{\calM_\lambda}(\vec{\beta})f^{\calL_s}(\vec{\beta}) = 0
\]

Let 
\[ \calF_s(\vec{X})  \stackrel{def}{=} f^{\calL_\lambda}(\vec{X})f^{\calM_s}(\vec{X}) - \alpha_s f^{\calM_\lambda}(\vec{X})f^{\calL_s}(\vec{X}) \in \FF_{q^m}[\vec{X}].
\] 
The polynomial $ \calF_s$ has degree $q^{\lambda+1}+q^{\lambda}+2 \sum\limits_{j=1}^{\lambda-1}q^j +1-q^{\lambda-s}$

This gives us a multivariate polynomial system over $\FF_{2^m}$ for which $\vec{\beta}$ is a solution. However, from our hypotheses we can do better and even reduce the degrees of the polynomials. 

Since $\beta_1,\ldots,\beta_\lambda$ are linearly independent they cannot be roots of linear factors over $\FF_q$ of $\calF_s$. Therefore
 we can reduce for all $s$ the polynomial $\calF_s(\vec{X})$ by its $\FF_q$-linear factors.  

\begin{lemma} Let us define 
\[
f_0(\vec{X}) = \prod\limits_{a \in \FF_q} (X_1+a)\prod\limits_{i=2}^{\lambda-1}\left(\prod\limits_{a_0,\dots a_{i-1} \in \FF_q}(X_i + \sum\limits_{j=1}^{i-1}a_jX_j +a_0)\right)
\]

 For any set $\calI$ of cardinality $\lambda$,  $f^{\calI}(\vec{X})$ is divisible by $f_0(\vec{X})$
\end{lemma}

\begin{proof}
Let  $\vec{\beta}$ be a root of $X_i + \sum\limits_{j=1}^{i-1}a_jX_j +a_0$ then they are $\FF_q$ co-linear. Hence, for all set of cardinality $\lambda$, $\calI$, the corresponding columns of $Mat^\calI(\vec{\beta})$ are co-linear. Therefore, $f^{\calI}(\vec{\beta}) = \det(Mat^{\calI}(\vec{\beta})) = 0$
\end{proof}

We have the following two corollaries

\begin{corollary}
We have $f^{\calJ_\lambda -1}(\vec{X}) = f_0(\vec{X})$     
\end{corollary}

\begin{proof}
Both polynomials are monic. Since $\calJ_{\lambda}-1 = \{0,\ldots, \lambda-1\}$, they also have the same degree $\sum\limits_{i=0}^{\lambda - 1}q^i$ 
\end{proof}

\begin{corollary}
\label{Cor:F_vec}
   For all    $\calI = \{i_1,\dots,i_\lambda\}$, we have  $\left(f_0(\vec{X})\right)^{[i_1]} | f^\calI(\vec{X})$
\end{corollary}

We have
\begin{itemize}
  \item From the lemma: $f_0(\vec{X})$ divides $f^{\calM_s}(\vec{X})$ and $f^{\calL_s}(\vec{X})$ for all $s \in \{1,\ldots,\lambda-1\}$.
  \item From corollary \ref{Cor:F_vec} : $f_0(\vec{X})^{[1]}$ divides $f^{\calM_\lambda}(\vec{X})$ and  $f^{\calL_\lambda}(\vec{X})$, since the minimum index of the sets is equal to $1$.
\end{itemize}
Therefore, for all $s \in \{1,\ldots,\lambda-1\}$, $\calF_s(\vec{X})$ can be divided by $f_0(\vec{X})^{q+1}$. We now consider the reduced polynomials  
\[
 \forall   s \in \{1,\ldots,\lambda-1\}, \quad \calP_s(\vec{X}) \stackrel{def}{=} \dfrac{\calF_s(\vec{X})}{(f_0(\vec{X}))^{q+1}} = \dfrac{\calF_s(\vec{X})}{
f^{\calJ_\lambda-1}(\vec{X})f^{\calJ_\lambda}(\vec{X})}
\]
This gives us a new polynomial system for which $\vec{\beta}$ is also a solution, but the degree is reduced.

\begin{lemma}
Let $\vec{A} = (a_{i,j})_{i=0,j=0}^{\lambda-1,\lambda-1} \in \textbf{PGL}(\lambda; \FF_q)$. Consider the tranformation on  $f^{\calI}(\vec{X})$ defined on $\vec{X}=(X_1,\ldots, X_{\lambda-1})$ by 
\[
\forall j \in \{1,\ldots,\lambda-1\}, \quad X_j \longmapsto \dfrac{a_{0,j} +\sum\limits_{i=1}^{\lambda-1}a_{i,j}X_i}
{a_{0,0} + \sum\limits_{i=1}^{\lambda-1}a_{i,0}X_i}
\]
then the polynomial $f^{\calI}(\vec{X})$ is transformed into
\[
 f^\calI(\vec{X}) \longmapsto  \vec{A}.f^{\calI}(\vec{X}) \stackrel{def}{=} 
   \dfrac{\Delta_A}{(a_{0,0} + \sum\limits_{i=1}^{\lambda-1}a_{i,0}X_i)^{\deg(f^{\calI})}}f^\calI(\vec{X})
\]
where $\Delta_A$ is the determinant of $\bfA$.
\end{lemma}

\begin{proof}

    Let $D = a_{0,0} + \sum\limits_{i=1}^{\lambda-1}a_{i,0}X_i$.
    Thus, for $j=1,\dots,\lambda$, the $j$th row of  $Mat^\calI(\vec{X})$ denoted by $Row_j(Mat^\calI(\vec{X}))$
    becomes
     \[  
       Row_j(Mat^\calI(\vec{X})) \longmapsto \dfrac{Row_j\left(Mat^\calI(\vec{X}) \cdot \bfA\right)}{D^{[i_j]}}.
   \]
   Therefore, since $\deg(f^{\calI}) = \sum_{j \in \calI}{[j]}$, from lemma \ref{Lemma:PolyMult}, we obtain 
   
   \begin{align*}
       \det(Mat^\calI(\vec{X})) &\longmapsto \dfrac{ \det \vec{A}}{D^{\deg(f^{\calI})}}\det(Mat^\calI(\vec{X}))\\
       f^\calI(\vec{X}) &\longmapsto \dfrac{\Delta_A}{D^{\deg(f^{\calI})}} f^\calI(\vec{X})
   \end{align*}
\end{proof}
Apply the lemma, we have

\begin{align*}
    \calF_s(\vec{X}) &\longmapsto \dfrac{\Delta_A^2}{
D^{q^{\lambda+1}+q^{\lambda}+2\sum\limits_{j=1}^{\lambda-1}q^j+1-q^{\lambda-s}}}\calF_s(\vec{X})\\
    f^{\calJ_\lambda-1}(\vec{X}) &\longmapsto \dfrac{\Delta_A}{
D^{\sum\limits_{j=0}^{\lambda-1}q^j}} f^{\calJ_\lambda-1}(\vec{X})\\
    f^{\calJ_\lambda}(\vec{X})  &\longmapsto \dfrac{\Delta_A}{
D^{\sum\limits_{j=1}^{\lambda}q^j}} f^{\calJ_\lambda}(\vec{X}) 
\end{align*}

Hence,
\[
    \calP_s(\vec{X}) \mapsto \dfrac{1}{
D^{q^{\lambda+1}-q^{\lambda-s}}} \calP_s(\vec{X})
\]

We therefore have 
\begin{proposition}
\label{prop:stable}
If there isn't any common factor between the polynomials $\calP_s(X)$, then the set of root of the polynomial system 
\begin{equation}
\label{Eq:SystBeta}
\forall i =1,\ldots,  \lambda-1,\quad \calP_i(\vec{X}) = 0
\end{equation}
equals the orbit of any root under the group action of $\textbf{PGL}(\lambda, \FF_q)$
\end{proposition}

\begin{proof}
If there isn't any common factor between the polynomials $\calP_s(X)$ then the number of roots is at bounded by $\prod\limits_{j=1}^{\lambda-1}(q^{\lambda+1} -q^j) = |\textbf{PGL}(\lambda,q)|$ (Bezout bound \cite{FGHR13}). Moreover, any element in the orbit of a solution $\beta$ under the group action of $\textbf{PGL}(\lambda, \FF_q)$ is again root of the system. From Assumption \ref{asp:3}  the orbit of $\beta$ under $\textbf{PGL}(\lambda, \FF_q)$ has cardinality $= |\textbf{PGL}(\lambda, \FF_q)| $ which means that the stabilization
of $\beta$ with respect to this group action is trivial. In that case any root of the system (\ref{Eq:SystBeta}) corresponds to an element of $\mathcal{S}$.
\end{proof}

For instance, when $q=2$ and $\lambda =3$ the system of equation below taking ($\beta_1,\beta_2$) as solution:

\[
      \begin{cases} 
   Pr_1(X,Y) = 0 \\
   Pr_2(X,Y) = 0
  \end{cases}
\]
 
This is a system of 2 polynomial equation in 2 variables. In practice, by using MAGMA, we can see that there isn't any common factor between $Pr_1(X)$ and $Pr_2(X)$. Therefore, the number of roots has Bezout 's upper bound by the product of the degrees of $Pr_1(X,Y)$ and $Pr_2(X,Y)$. 

Therefore, the number of roots are at most $(q^4-q)(q^4-q^2) = |\textbf{PGL}(3,q)|$. Thus, all the roots are in the orbit of a root under an action of $\textbf{PGL}(3,q)$

The remaining problem is finding a root of the system of equation above. It can be done by the following steps:

\begin{enumerate}
    \item Calculating $Res(Pr_1, Pr_2, Y)$ the resultant of $Pr_1$ and $Pr_2$ in the variable $Y$. We obtain a univariate polynomial of degree 168 in variable $X$. Finding one root $x_0$ of this polynomial.
    \item Calculating $\gcd(Pr_1(x_0,Y),Pr_2(x_0,Y))$ which is a polynomial of degree 4 in variable Y. Taking one root $y_0$ and verify it is a root of the system of equation. 
\end{enumerate}

In general, the problem of finding one root of a system of polynomial equation is a hard question as well as finding all roots of  a system of polynomial equation. \\

\textbf{Polynomial System Solving over Finite Fields}
Let $\FF$ is a finite field.
\textbf{Input:} $f_1(x_1,..., x_n),..., f_m(x_1,..., x_n) \in \FF[x_1,..., x_n]$.

\textbf{Goal:} Find a vector $\alpha = (\alpha_1,...,\alpha_n) \in \FF^n$
s.t:
$f_1(\alpha) = \dots = f_m(\alpha) = 0.$\\

Theoretically it is a NP-hard problem (problem AN9 p.251 in Appendix: A list of NP-complete problem \cite{GJ79}). For the special cases $\lambda=2$ and $\lambda =3$ finding a solution can be done in polynomial-time (by using properties of resultants for $\lambda=3$).

However, in the case of no common factor, the number of roots is bounded by Bezout bound. To check that $Pr_i(X)_{i=1}^{\lambda-1}$ don't have common factor, we can check whether $Res(Pr_i(X),Pr_j(X),{X_1}) \ne 0,~\forall 1 \le i <j \le \lambda-1$. (Prop. 1, Ch. 3, \cite{CLO07}).
It costs $O(d^3)$ where $d = \prod\limits_{j=1}^{\lambda-1}(q^{\lambda+1} -q^j)$ arithmetic operations over $\FF_{q^m}[X_2,\dots,X_{\lambda-1}]$.

 We can see the importance of the Assumption\ref{asp:3} in the Proposition~\ref{prop:stable}. In the case where this assumption does not satisfy, i.e there exists $\overline{\bfA} \in \textbf{PGL}(\lambda,\FF_q) \backslash \overline{\bfI_{\lambda}}$ and $\bfA=  (a_{ij})_{i,j=1}^{\lambda} $ such that

\[
    \beta_j = \dfrac{a_{0,j} +\sum\limits_{i=1}^{\lambda-1}a_{i,j}\beta_i}{a_{0,0} + \sum\limits_{i=1}^{\lambda-1}a_{i,0}\beta_i}
\]

Thus, $\beta$ is a root of a system of $\lambda-1$ polynomial equations of degree 2:

\[
    (a_{0,0} + \sum\limits_{i=1}^{\lambda-1}a_{i,0}X_i)X_j - (a_{0,j} +\sum\limits_{i=1}^{\lambda-1}a_{i,j}X_i) = 0, ~j = 1,\dots,\lambda-1
\]

This polynomial is different from 0. Indeed, if it was, $a_{0,0} = a_{j,j}$ for $j = 1,\dots,\lambda-1$ and $a_i,j = 0$ for $i \ne j$, which means $\bfA \in \overline{I_{\lambda}}$.

This system is multivariate quadratic (MQ)-system,  the associated problem to decide if this system is solvable or not, also known as MQ-problem, is proven to be NP-complete \cite{GJ79}. Some algorithms used to solve this system is reviewed in the paper \cite{TW10}. In case of $\lambda = 3$, this can be solved easily by Resultant. Therefore, when the Assumption~\ref{asp:3} does not satisfy, we can exploit some information about $\beta$ by solving a multivariate quadratic system.

\subsubsection{Final step:}
Now from a solution $\vec{\beta}$ to (\ref{Eq:SystBeta}), we aim at finding the corresponding vector $\vec{h} \in (\FF_{q^m}^n)^{\lambda}$ such that $(\vec{h},\vec{\beta}) \in \mathcal{S}$.

We point out the key steps in the Coggia-Couvreur attack for $\lambda$ as follows. To be convenient, we denote known elements by blue color and unknown elements by red color. Given ${\color{blue}\beta_1',\dots,\beta_{\lambda-1}'}$, recover $({\color{red}\bfh_0',\dots,\bfh_{\lambda-1}'},{\color{blue}\beta_1',\dots,\beta_{\lambda-1}'})$ corresponding.

\begin{enumerate}

    \item For $\calI = \{1,\dots,\lambda\}$, since ${\color{blue}\vec{\beta'}}$ is known, $k_i = \dfrac{f^{-(\calI \setminus \{ i\} ) \cup \{ 0 \}}({\color{blue}\vec{\beta'}})}{ f^{-\calI}({\color{blue}\vec{\beta'}})}, i = 1,\dots,\lambda$ can be computed. Moreover, from the Lemma \ref{Lemma:Sum}, there exists a unique $\lambda$-tuple ${\color{blue}\bfu_\calI} = ({\color{blue}\bfu_1},\dots,{\color{blue}\bfu_\lambda}) \in \bigtimes\limits_{j=1}^{\lambda}{\color{blue}\calA_j}$ such that $\sum_{i=1}^{\lambda} {\color{blue}\bfu_i} = \color{blue}\bfu_0$, so we can compute
    
    ${\color{red}\bfh_0'}+\sum\limits_{j=1}^{\lambda-1}{\color{blue}\beta_j'}^{[-i]}{\color{red}\bfh_j'} = \dfrac{{\color{blue}\bfu_i}}{{\color{blue}k_i}}, i = 1,\dots,\lambda$. 
    
        \[
    ({\color{red}\bfh_0',\dots,\bfh_{\lambda-1}'})\begin{bmatrix}
  ~ & 1 & 1 &\ldots & 1 \\ 
  ~ & {\color{blue}\beta_1^{[-1]}} & {\color{blue}\beta_1^{[-2]}}  &\ldots &{\color{blue}\beta_1^{[-\lambda]}}\\ 
  ~ & \vdots &  \vdots &\ddots &\vdots\\
  ~ & {\color{blue}\beta_{\lambda-1}^{{[-1]}}} & {\color{blue}\beta_{\lambda-1}^{{[-2]}}} &\ldots &{\color{blue}\beta_{\lambda-1}^{{[-\lambda]}}}
\end{bmatrix} = \left(\dfrac{{\color{blue}\bfu_1}}{{\color{blue}k_1}},\dfrac{{\color{blue}\bfu_2}}{{\color{blue}k_2}},\dots,\dfrac{{\color{blue}\bfu_{\lambda}}}{{\color{blue}k_{\lambda}}} \right)
    \]
    It implies to a linear system of $\lambda$ equation and $\lambda$ unknowns which are vectors ${\color{red}\bfh_0'},\dots,{\color{red}\bfh_{\lambda-1}'}$ and the determinant of the matrix of coefficients is non-zero.
    \item After recovering an alternate key of the form $({\color{red}\bfh_0',\dots,\bfh_{\lambda-1}'},{\color{blue}\beta_1',\dots,\beta_{\lambda-1}'})$, we can compute the dual code $\cpubp$ and hence decrypt the ciphertext.
\end{enumerate}

\subsection{Complexity of the case $\lambda =3$}

This part shows the complexity of the attack by giving the number of operation in $\FF_{q^m}$. Let $\omega$ be the exponent of the complexity of linear algebra operations. The Frobenius map costs $O(\log q)$ operations.

\textbf{Step 1.} 
\begin{itemize}
    \item Computation of dual code $\cpubp$ cots $O(n^\omega)$ operations.
    \item Computation of $\cpubp^{[i]}, \forall i = 1,\dots,n-k+1$ costs $O(n^2\log q)$ operations.
    \item Computation $S_j = \sum\limits_{i=j}^{j+\lambda-1}\cpubp^{[i]}$ uses Gaussian elimination, so it costs $O(n^\omega)$. Thus, computation $\bigcap_{i=0}^{n-k-\lambda+1}S_j$ costs $O(n^{\omega+1})$.
    
    Overall step 1 costs $O(n^3\log q + n^{\omega+1})$ operations.
\end{itemize}

\textbf{Step 2.}
\begin{itemize}
    \item Computation $(u_1^\calI,\dots,u_\lambda^\calI)$ represents the resolution of a linear system $\lambda$ unknowns and $n$ equations costs $O(n)$ operations. This computation performed $O(n)$ times, so it costs $O(n^2)$ operations.
    \item Complexity of finding a root of a polynomial of degree $\tilde{d}$ by Cantor–Zassenhaus algorithm (\cite{bcg17}) costs $\tilde{O}(\tilde{d}^2m \log q)$ operations in $\FF_{q^m}$ for $\tilde{d} = (q^4-q)(q^4-q^2)$.
    \item Computation of resultant of bivariate polynomials $Res(P_1,P_2,X)$ which $P_1, P_2$ of degree $d,e$ by Lickteig–Roy subresultant algorithm costs $O(d^2e)$ (\cite{lec19}).
    \item A finite number of linear systems solving costs $O(n^\omega)$
\end{itemize}

\textbf{Summary.} For $m = O(n)$, overall cost of $O(n^3\log q + n^{\omega+1}) + \tilde{O}(d^2n \log q)$ for $d = (q^4-q)(q^4-q^2)$.

\subsubsection*{Conclusion}
We provided a distinguisher for the Loidreau's scheme for any $\lambda$ and the public code has rate $R_{pub} \ge 1 - 1/ \lambda$. From this distinguisher, we are able to complete a polynomial time key recovery with the assumption of finding one root of a system of polynomial equation. Moreover, we have extended the key recovery attack for $\lambda =3$.

The parameters of $(k,n)$, which $R_{pub} \ge 1 - 1/ \lambda$ should be avoided in Loidreau's scheme. In the future, it will be worthwhile to attempt a modification of the attack to work for lower rate codes $R_{pub} < 1 - 1/ \lambda$ as well.

\bibliographystyle{ieeetr}
\bibliography{biblio}

\begin{thebibliography}{10}

\bibitem{Eliece78}
R.~J. {McEliece}, ``{A Public-Key Cryptosystem Based On Algebraic Coding
  Theory},'' {\em Deep Space Network Progress Report}, vol.~44, pp.~114--116,
  Jan. 1978.

\bibitem{gpt91}
E.~M. Gabidulin, A.~V. Paramonov, and O.~V. Tretjakov, ``Ideals over a
  non-commutative ring and their application in cryptology,'' in {\em Advances
  in Cryptology --- EUROCRYPT '91} (D.~W. Davies, ed.), (Berlin, Heidelberg),
  pp.~482--489, Springer Berlin Heidelberg, 1991.

\bibitem{gmrz13}
P.~Gaborit, G.~Murat, O.~Ruatta, and G.~Zemor, ``{Low Rank Parity Check codes
  and their application to cryptography},'' in {\em {The International Workshop
  on Coding and Cryptography (WCC 13)}} (L.~Budaghyan, T.~Helleseth, and M.~G.
  Parker, eds.), (Bergen, Norway), p.~13 p., Apr 2013.
\newblock ISBN 978-82-308-2269-2.

\bibitem{over08}
R.~Overbeck, ``Structural attacks for public key cryptosystems based on
  gabidulin codes,'' {\em J. Cryptology}, vol.~21, pp.~280--301, 2008.

\bibitem{Loi17}
P.~Loidreau, ``{A new rank metric codes based encryption scheme},'' in {\em
  {PQCrypto 2017}} (T.~Lange and T.~Takagi, eds.), vol.~10346 of {\em Lecture
  Notes in Computer Science}, (Utrecht, Netherlands), pp.~3--17, {Springer},
  June 2017.

\bibitem{CC19}
D.~Coggia and A.~Couvreur, ``On the security of a loidreau's rank metric code
  based encryption scheme,'' {\em CoRR}, vol.~abs/1903.02933, 2019.

\bibitem{gha20}
A.~Ghatak, ``Extending coggia-couvreur attack on loidreau's rank-metric
  cryptosystem,'' 2020.

\bibitem{FGHR13}
J.-C. Faug{\`e}re, P.~Gaudry, L.~Huot, and G.~Renault, ``{Polynomial Systems
  Solving by Fast Linear Algebra}.'' 27 pages, Apr. 2013.

\bibitem{GJ79}
M.~R. Garey and D.~S. Johnson, {\em Computers and Intractability; A Guide to
  the Theory of NP-Completeness}.
\newblock USA: W. H. Freeman Co., 1990.

\bibitem{CLO07}
D.~A. Cox, J.~Little, and D.~O'Shea, {\em Ideals, Varieties, and Algorithms: An
  Introduction to Computational Algebraic Geometry and Commutative Algebra, 3/e
  (Undergraduate Texts in Mathematics)}.
\newblock Berlin, Heidelberg: Springer-Verlag, 2007.

\bibitem{TW10}
E.~Thomae and C.~Wolf, ``Solving systems of multivariate quadratic equations
  over finite fields or: From relinearization to mutantxl,'' 2010.

\bibitem{bcg17}
A.~Bostan, F.~Chyzak, M.~Giusti, R.~Lebreton, G.~Lecerf, B.~Salvy, and
  {\'E}.~Schost, {\em Algorithmes Efficaces en Calcul Formel}.
\newblock Palaiseau: Fr{\'e}d{\'e}ric Chyzak (auto-édit.), Sept. 2017.
\newblock 686 pages. Imprim{\'e} par CreateSpace. Aussi disponible en version
  {\'e}lectronique.

\bibitem{lec19}
G.~Lecerf, ``{On the complexity of the Lickteig-Roy subresultant algorithm}.''
  working paper or preprint, Jan. 2017.

\end{thebibliography}

\end{document}